\documentclass[11pt]{llncs}
\def\draft{0}
\def\anonymous{0} 
\usepackage[margin=1in]{geometry}
\usepackage[T1]{fontenc}
\usepackage{graphicx} 
\usepackage{cite}
\usepackage{enumerate}
\usepackage{amsmath,amssymb,amsfonts}
\usepackage{booktabs}
\usepackage{braket}
\usepackage{algorithm}
\usepackage{algpseudocode}
\usepackage{hyperref}
\usepackage{cleveref}
\usepackage{color}
\usepackage{xcolor}

\newcommand{\hide}[1]{}

\def\clap#1{\hbox to 0pt{\hss#1\hss}}

\newcommand\textunderset[2]{%
  \leavevmode
  \vtop{\offinterlineskip
    \halign{%
      \hfil##\hfil\cr 
      \strut#2\cr
      \noalign{\kern-.3ex}
      #1\strut\cr
    }%
  }%
}
\newcommand\textoverset[2]{%
  \leavevmode
  \vbox{\offinterlineskip
    \halign{%
      \hfil##\hfil\cr 
      \dynscriptsize\strut#1\cr
      \noalign{\kern-.3ex}
      #2\strut\cr
    }%
  }%
}

\providecommand{\deq}{\mathrel{:=}}

\providecommand{\getsr}{\stackrel{\smash{\$}}\gets}

\makeatletter
\def\metadef#1#2{%
  \def\metadef@iter##1{\ifx##1;\else \expandafter\newcommand\csname#1\endcsname{#2}\expandafter\metadef@iter\fi}%
  \expandafter\metadef@iter%
}
\makeatother

\metadef{vec#1}{\mathbf #1}abcdefghijklmnopqrstuvwxyz;
\metadef{vec#1}{\mathbf #1}ABCDEFGHIJKLMNOPQRSTUVWXYZ;

\metadef{mat#1}{\mathbf #1}ABCDEFGHIJKLMNOPQRSTUVWXYZ;

\metadef{frak#1}{\mathfrak #1}ABCDEFGHIJKLMNOPQRSTUVWXYZ;

\metadef{c#1}{\mathcal #1}ABCDEFGHIJKLMNOPQRSTUVWXYZ;

\metadef{s#1}{\mathsf #1}ABCDEFGHIJKLMNOPQRSTUVWXYZ;

\providecommand{\ZZ}{\mathbb Z}

\makeatletter
\@ifpackageloaded{mathtools}{}{
  \def\clap#1{\hbox to 0pt{\hss#1\hss}}

}
\makeatother

\newcommand{\negl}{\mathop{\operatorname{negl}}}
\newcommand{\poly}{\mathop{\operatorname{poly}}}
\newcommand{\bit}{\{0,1\}}

\providecommand{\fn}[1]{\ifmmode\operatorname{\mbox{\textnormal{#1}}}\else\mbox{\textnormal{#1}}\fi}
\newcommand{\sffn}[1]{\ifmmode\operatorname{\mbox{\textnormal{\textsf{#1}}}}\else\mbox{\textnormal{\textsf{#1}}}\fi}

\newcommand{\Sim}{\textnormal{\textsf{Sim}}}

\providecommand{\tk}{\textnormal{\textsf{tk}}}

\ifnum\draft=1
\newcommand{\bnote}[1]{%
  \textcolor{brown}{\textbf{[Boyang:} #1\textbf{]}}%
}
\newcommand\luojian[1]{\textcolor{cyan}{Luojian: #1}}
\else
\newcommand{\bnote}[1]{}
\newcommand{\luojian}[1]{}
\fi

\begin{document}

\title{
Beyond NISQ Assumptions: One-time Memory in the Classically Accessible Random-Oracle Model
}
\ifnum\anonymous=1
\author{Anonymous Submission}
\institute{}
\else
\author{
  Boyang Chen\inst{1} \and
  Tianren Liu\inst{2} \and
  Luojian Wei\inst{2}
}
\authorrunning{Boyang Chen, Tianren Liu, and Luojian Wei}
\institute{
  Department of Computer Science and Technology, Tsinghua University, Beijing, China
  \and
  Center on Frontiers of Computing Studies, Peking University, Beijing, China
}
\fi
\maketitle

\begin{abstract}
  Quantum information enables many cryptographic primitives that are
  impossible in the classical world.  A line of works has developed
  cryptographic protocol under the assumption that quantum adversaries
  are restricted to noisy intermediate-scale quantum (NISQ) computing
  power, enabling strong one-time functionalities. But the advent of
  early fault-tolerant quantum computers eras will allow deeper
  logical quantum circuits, calling into questions the applicability
  of these NISQ-based assumptions.

  In this work, we adapt the classically accessible random oracle
  model (CAROM) as in~\cite{AC:BDFLSZ11} and~\cite{EPRINT:AnaKal22},
  in which adversaries are only allowed to classically query the
  random oracle.  The restriction is well motivated for NISQ quantum
  adversaries and may remain plausible in the presence of early
  fault-tolerant quantum computers.  Then, we show that an efficient
  simulation-secure one-time memory (OTM) is possible under CAROM. Our
  protocol uses only BB84 states and has quadratic communication: for
  \(\lambda\)-bit message and integer-valued parameters $n=n(\lambda)$
  and $\ell =\ell(\lambda)$, the construction uses $n \ell$ qubits and
  $(n+2)\lambda$ classical bits, and for any quantum adversary with at
  most $2^{\ell/2-1}-1$ classical queries to the random oracle, its
  simulation advantage is at most
  $$ (n+3) \left(\frac 34\right)^n. $$
  Hence exponentially small simulation advantage in $n$.

\end{abstract}

\section{Introduction}
\label{sec:introduction}

Quantum information enables many classically impossible cryptographic
applications. Early milestones include Wiesner's quantum
money~\cite{wiesner1983conjugate} and the BB84 quantum
key-distribution protocol~\cite{C:BenBra84}. Recently, a
variety of exciting applications have emerged, including unclonable banknotes
with public
verification~\cite{STOC:AarChr12,JC:Zhandry21,EC:LiuMonZha23,STOC:BosNehZha25},
copy-protection of
programs~\cite{C:ALLZZ21,C:CLLZ21} and one-time
programs~\cite{C:BroGutSte13,STOC:GLRRV25}, encryption
with certified
deletion~\cite{TCC:BroIsl20,C:BarKhu23,EC:BGKMRR24,AC:HMNY21},
and untelegraphable encryption~\cite{TCC:CKNY25}.

However, several quantum cryptographic functionalities are impossible
in the standard model. For example, one-time memory (OTM), proposed by
Goldwasser, Kalai and Rothblum~\cite{C:GolKalRot08} is a
non-interactive form of 1-out-of-2 oblivious transfer, which is
impossible in the standard model.  In the one-time memory protocol, a
sender prepares a token holding two messages $m_0$ and $m_1$.  A
receiver chooses a bit $b$, obtains $m_b$, and should learn nothing
more about $m_{1-b}$.  Evaluation consumes the token.  This simple
interface is a building block for one-time programs and one-time
proofs~\cite{C:GolKalRot08,C:BroGutSte13,STOC:GLRRV25}: it
allows access to a program or a proof to be limited to a single use.

Restrictions of adversary's computational resources are introduced to
bypass the impossibility. Following this approach, a line of works
construct efficient one-time memory protocols under various
computation-bounded models.
The isolated-qubits model restricts the adversary to local operations
and classical communication~\cite{ITCS:Liu14,C:Liu14}.  Depth-bounded
constructions combine limits on quantum computation with time-lock
puzzles or
obfuscation~\cite{ITCS:Liu23a,stambler2025semquantum,stambler2026simple}.
However, all these works impose a very strong assumption: the
adversary can either implement only very shallow circuits, or has no
storage at all. All these restrictions on the adversary fall into
the near-term intermediate-scale quantum computing (NISQ) paradigm, in
the sense that the adversary is always characterized by shallow
quantum circuits.

However, as we enter the era of early fault-tolerant quantum
computing, such restrictions may become increasingly difficult to
justify. This naturally raises the following question:

\begin{center}
  \emph{Can we design an efficient one-time memory protocol with a more
    robust assumption?}
\end{center}

Our answer is affirmative: we adapt the model
called the \emph{classically accessible random-oracle model} (CAROM),
which is arguably safe even against early fault-tolerant computers,
and construct one-time memory in it. Several papers consider the
classically accessible random-oracle model and contrast it with
quantum oracle access, including~\cite{AC:BDFLSZ11, EC:YamZha21,
  ITCS:LLPY24}. \cite{EPRINT:AnaKal22} formalizes the classically accessible
model and names it CAROM. In CAROM, a quantum adversary may perform
joint quantum operations on the token and its workspace, but may query
the random oracle only on classical inputs.

\begin{definition}[CAROM, informal]
  Let $\cO:\mathcal{D}\to\bit^\lambda$ be a uniformly random function
  on a classical query domain $\mathcal{D}$.  A $q$-query CAROM
  adversary may perform arbitrary quantum computation, retain quantum
  memory, and use joint or entangled operations between oracle calls.
  It may make at most $q=q(\lambda)$ adaptive queries in total, and
  every query to $\cO$ must be a classical string.
\end{definition}

Thus, CAROM captures a more general restriction: the adversary
need not have shallow depth or short coherence
time; the only requirement is that it
cannot coherently evaluate the random oracle. CAROM can be
instantiated with a hard public hash function in the real world. 
Prior work suggests that replacing the random oracle with SHA-2 or SHA-3
might be a plausible instantiation under the random-oracle heuristic~\cite{STOC:ACCGSW23}. 
Even for early fault-tolerant quantum computers, this may still
remain plausible, since coherent evaluation is plausibly hard for
state-of-the-art quantum computers.

In the CAROM, we can construct statistically secure one-time memory.

\begin{theorem}[Main theorem, informal]
  For any \(\lambda \in \ZZ\), there is an OTM scheme for message
  length \(\lambda\), with communication of $n\ell$ qubits and
  $n\lambda+2\lambda$ classical bits. Its correctness error is at most
  $n2^{-\lambda}$.  For an arbitrary strategy making at
  most $q$ adaptive classical random-oracle queries, the 
  security loss is at most
  \[
    2^{1-n}+(n+1)\left( \frac{1}{2}+\frac{q+1}{2^{\ell/2+1}}
    \right)^{n}.
  \]
  where $n=n(\lambda)\geq 1$ and $\ell=\ell(\lambda)\geq 1$ are
  functions of $\lambda$. In particular, if
  $q\le 2^{\ell/2-1}-1$, the security loss is at most
  $(n+3) \left(\frac 34\right)^n$.
\end{theorem}

In particular, if we want to encode a $\lambda$-bit message in the
one-time memory, with security $2^{-\Omega(\lambda)}$ against
adversaries with at most $2^{O(\lambda)}$ classical queries to the
oracle, $O(\lambda^2)$ qubits of communication are needed; if we
furthermore restrict the adversary to at most $\poly(\lambda)$ oracle
queries, then setting $\ell=\poly\log\lambda$ and
$n=\lambda$ is enough to reach $2^{-\Omega(\lambda)}$ simulation
security.

\subsection{Technical Overview}

\paragraph{The first attempt.}  Recall that a one-time memory allows
the sender to encode two messages $m_0,m_1$ in a token, so that the
receiver choosing a bit $b$ can decode $m_b$ but learns nothing more
about $m_{1-b}$.  As a straightforward attempt, consider encoding two
secret keys $k_0$ and $k_1$ in BB84 states and masking each message
$m_b$ with $h(k_b)$ for a public universal hash function $h$.  
More formally, for every $i\in[\lambda]$ the sender
samples a basis bit $\theta_i\getsr\bit$ and a bit
$x_i\getsr\bit$, and prepares
$\ket{x_i}_{\theta_i}\deq H^{\theta_i} \ket{x_i}$, so that
basis $0$ is the computational (Z) basis and basis $1$ is the Hadamard
(X) basis.

With the basis information the receiver can decode the key and hence
the message.  Let $S_b=\{i\in[\lambda]:\theta_i=b\}$, let $k_b$ be the
concatenation of the bits $x_i$ with $i\in S_b$, and set
$c_b=m_b\oplus h(k_b)$.
The sender sends the quantum registers $\ket{x_i}_{\theta_i}$ together
with the classical data $(S_0,S_1,c_0,c_1)$.  To decode $m_b$, the
receiver measures every register in basis $b$, concatenates the
outcomes with $i\in S_b$ to obtain $k_b$, and outputs
$m_b=c_b\oplus h(k_b)$.  The full na\"ive protocol is described
in~\Cref{alg:naive-semihonest-otm}.

This protocol is secure against semi-honest adversaries: to decode
$m_b$, a semi-honest adversary measures all the states in basis $b$,
which collapses the words with $\theta_i=1-b$ and destroys the
information about the words $x_i$ for all $i\in S_{1-b}$, thereby
consuming the information needed to reveal $m_{1-b}$.

\begin{algorithm}[t]
  \caption{Na\"ive semi-honest one-time memory protocol
  \label{alg:naive-semihonest-otm}}
  \begin{minipage}[t]{0.49\textwidth}
    \textbf{Sender: $\mathsf{Token.Gen}(1^\lambda,m_0,m_1)$}
    \begin{algorithmic}[1]
      \State Sample $x_i\getsr\bit$ and
      $\theta_i\getsr\bit$ for all $i\in[\lambda]$.
      \State Prepare
      \[
        \ket{x_i}_{\theta_i}
        =
        \begin{cases}
          \ket{x_i}, & \theta_i=0,\\
          H\ket{x_i}, & \theta_i=1,
        \end{cases}
      \]
      for every $i\in[\lambda]$.
      \State Set
      \[
        S_0\gets\{i\in[\lambda]:\theta_i=0\},
        \qquad
        S_1\gets\{i\in[\lambda]:\theta_i=1\}.
      \]
      \State For each $b\in\bit$, let
      \[
        k_b\gets (x_i)_{i\in S_b}.
      \]
      \State Set
      \[
        c_b\gets m_b\oplus h(k_b),
        \qquad b\in\bit.
      \]
      \State Let
      \[
        \ket{\psi}\gets
        \bigotimes_{i=1}^{\lambda}\ket{x_i}_{\theta_i}
      \]
      and
      $
        c\gets(S_0,S_1,c_0,c_1).
      $
      \State Output $\ket{\tk}=(\ket{\psi},c)$.
    \end{algorithmic}
  \end{minipage}
  \hfill
  \begin{minipage}[t]{0.49\textwidth}
    \textbf{Receiver: $\mathsf{Token.Eval}(\ket{\tk},b)$}
    \begin{algorithmic}[1]
      \State Parse $\ket{\tk}$ as $\ket{\psi}$ and
      \[
        c=(S_0,S_1,c_0,c_1).
      \]
      \State Measure every register of $\ket{\psi}$ in basis
      \[
        B_b=
        \begin{cases}
          Z,&b=0,\\
          X,&b=1.
        \end{cases}
      \]
      \State Let $x'_1,\ldots,x'_\lambda\in\bit$ be the
      measurement outcomes.
      \State Set
      \[
        k'_b\gets(x'_i)_{i\in S_b}.
      \]
      \State Output
      \[
        m'_b\gets c_b\oplus h(k'_b).
      \]
    \end{algorithmic}
  \end{minipage}
\end{algorithm}

However, this is clearly not secure against malicious adversaries: as
the basis information is sent to the receiver in the clear, the
adversary can simply measure the registers with $i\in S_0$ in the Z
basis and those with $i\in S_1$ in the X basis, and decode both keys
$k_0,k_1$, hence both messages $m_0,m_1$.

\paragraph{Introducing the CAROM.}  Indeed, a more fundamental
impossibility rules out one-time memories in the standard model.
Suppose a quantum adversary can recover $m_b$ with high probability.
Decoding $m_b$ can then be regarded as a measurement on the encoded
state that succeeds with high probability, so by the gentle
measurement lemma~\cite{Win99gentle} it disturbs the state only a
little; consequently, the adversary can decode both $m_0$ and $m_1$
with high probability, violating the security requirement.

To bypass this impossibility, we must force the adversary to
destructively measure the state during decoding.  We do so in an
unstructured idealized model---the classically accessible
random-oracle model (CAROM), in which the adversary can query the
random oracle only on classical inputs.  If we can arrange that the
encoded state must be measured before the basis information is revealed,
one-time security might hold.

But the output of the random oracle is uniformly random and
independent of the basis choices, so how can it be used to encode the
message?  The idea is to \emph{authenticate the correct measurement
  outcomes} with the random-oracle responses, called \emph{tags}.
More precisely, instead of sending the basis information $S_0,S_1$
directly to the receiver, the sender sends the tag
$\tau_i=\cO_{\mathrm{tag}}(i,x_i,\theta_i)$ for every word.  To
recover $m_b$, the receiver measures all the states in basis $b$ and,
for each outcome $x'_i$, checks whether $\cO_{\mathrm{tag}}(i,x'_i,b)$
equals $\tau_i$; the matching words identify the set $S_b$, and then
the receiver computes $k_b$ and decodes $m_b$ with the mask.

But this protocol also admits a straightforward attack.  Since every
word $x_i$ is a \emph{single bit}, the adversary can query
$\cO_{\mathrm{tag}}(i,y,\xi)$ for all $y,\xi\in\bit$ and compare the
responses with $\tau_i$, thereby decoding every $x_i$. In this case,
the CAROM is effectively queried coherently by enumeration. To rule
out this attack, we take the words to be \emph{$\ell$-bit
strings} instead: the sender samples $x_i\getsr\bit^\ell$ uniformly
and prepares
$\ket{x_i}_{\theta_i}=(H^{\theta_i})^{\otimes\ell}\ket{x_i}$, which
defeats enumeration by any adversary making fewer than $2^{O(\ell)}$
queries to $\cO_{\mathrm{tag}}(\cdot)$.

\paragraph{List-recovery security.}  To prove the security of the
one-time memory protocol in the CAROM, it suffices to show that
recovering both $k_0$ and $k_1$---equivalently, all the underlying words
$x_i$---is hard for adversaries with only classical oracle access.

Consider first the single-word setting. The adversary is given the BB84 state
$\ket{x}_{\theta}$ together with the tag
$\tau=\cO_{\mathrm{tag}}(x,\theta)$.  The
security then amounts to showing that the adversary cannot output the
correct $x$ with probability much better than $1/2$---the success
probability of the trivial strategy that guesses the preparation basis
and measures in it.

The key step of the security proof is to reduce the search for $x$ to a
list-decoding problem. Suppose that a quantum algorithm
$\cA^{\cO_{\mathrm{tag}}}$, making at most $q$ classical queries to
$\cO_{\mathrm{tag}}$ and given access to the BB84 state and the tag,
outputs $x$ with high
probability. We show that this implies the existence of an oracle-free algorithm
$\tilde{\cA}$ that, given only the BB84 state and the tag, outputs a
list of at most $q+1$ candidates containing $x$ with high probability.

The reduction is simple: the simulator $\tilde{\cA}$ runs $\cA$ with a
fresh random oracle $\tilde{\cO}_{\mathrm{tag}}$ and records all the
classical queries made by $\cA$ together with its final output. 
If $\cA$ never queries $x$, the simulation is perfect, since
$\cO_{\mathrm{tag}}(x')$ is independent of the tag
$\cO_{\mathrm{tag}}(x)$ for every $x'\neq x$, and the final output of
$\cA$ is then $x$; if $\cA$ does query $x$, that query is recorded as
well. Hence the simulator outputs a list containing $x$ with high
probability.
By a standard information-theoretic argument
(\Cref{lem:bb84-list-recovery-bound}), any algorithm that
list-recovers $x$ solely from $\ket{x}_\theta$ succeeds with
probability at most $\frac12+\negl(\lambda)$, which concludes the
security argument.

The general multi-word search bound can be derived with a parallel interactive proof argument---the security of the search game can be formalized as an interactive proof game. Then the multi-word search problem can be formalized as the parallel version of the interactive proof. By the standard parallel repetition theorem for interactive proofs (see, e.g., equation 4.53
  of~\cite{vidick2016quantum}), we can show that the multi-word search probability is exponentially small.

\subsection{Related Work}

\paragraph{Random-oracle constructions.}
Most closely related is Stambler's construction of simulation-secure
OTMs with classical random-oracle queries, using Wiesner blocks and
conjunction obfuscation~\cite{stambler2026simple}.  Its proof invokes
an extension of the obfuscation guarantee to quantum auxiliary inputs,
formulated as a conjecture. Our construction does not need the
conjunction obfuscation assumption.

\paragraph{Restrictions on quantum operations.}
Yi-Kai Liu's isolated-qubits constructions give information-theoretic
security against restricted one-pass measurements and, subsequently,
single-shot guarantees against general local operations and classical
communication~\cite{ITCS:Liu14,C:Liu14}.  Qipeng Liu instead studies
depth-bounded adversaries, combining Wiesner states with a
post-quantum weak time-lock puzzle~\cite{ITCS:Liu23a}.  Honest
generation and evaluation have constant quantum depth, and an
error-tolerant extension handles independent depolarizing noise at a
constant rate.  Stambler also studies geometrically local,
constant-depth adversaries \cite{stambler2025geometric}.  CAROM places
the restriction on oracle access rather than these quantum operations.


\paragraph{Classical oracle setups.}

Chung, Gerogiou, Lai and Zikas~\cite{CGLZ19} introduce the
classical-oracle one-time memory and constructs one-time memory
relative to a structured classical oracle from quantum disposable
MAC. Stambler's semi-quantum framework uses interactive setup with a
classical sender and an oracle that holds a master secret, decrypts
program ciphertexts, verifies token signatures, and evaluates
programs~\cite{stambler2025semquantum}. The framework also derives
OTMs with one-query simulation security in that setup.  Its oracle
functionality and token-generation protocol differ from our
sender-prepared BB84 token and random-function interface.

Several works construct the one-time memory protocol under the
assumption of stateless hardware, which can be viewed as a special
classical oracle. Broadbent, Gharibian, and Zhou give a
prepare-and-measure construction using trusted stateless hardware
queried classically~\cite{EPRINT:BroGhaZho18}.  Their statistical
security result is in the quantum universal composability framework,
subject to a fixed linear bound on queries in the number of token
qubits.  Stambler explores a different construction from stateless
hardware and quantum random-access codes~\cite{stambler2025qotm}.

All previous works rely on some structured oracle: the oracle
is tailored to verify the one-time memory token. In contrast, we build
one-time memory in the CAROM, which is an unstructured random oracle.

\paragraph{Related applications.}
Beyond the motivating one-time programs and proofs, prior work uses
garbling for secure outsourcing~\cite{AC:BelHoaRog12} and commodity
hardware for limited-use applications~\cite{TCC:EGGJZ22}.  Software
leasing, quantum copy protection, and one-time protection for
randomized programs study related ways to control access to
functionality
\cite{EC:AnaLaP21,C:ALLZZ21,STOC:GLRRV25,EPRINT:GunMov24}.  These
works provide broader context; our analysis concerns the OTM itself
and can be applied to construct one-time programs in the CAROM.

\subsection{Organization}

Section~\ref{sec:preliminaries} introduces the model, OTM definitions,
and quantum interactive-proof tools.  Section~\ref{sec:otm} presents
the construction and proves correctness and simulation security,
including the direct-product hardness bound.

\section{Preliminaries}
\label{sec:preliminaries}

\subsection{Notations}

We write $[n]=\{1,\ldots,n\}$.  The security parameter is
$\lambda$; the number of BB84 words $n=n(\lambda)\geq 1$ and their
length $\ell=\ell(\lambda)\geq 1$ are independent integer-valued
construction parameters.

\subsection{Operator Norm Theory}

\begin{lemma}[Projector sum bound]\label{lem:projector-sum}
  For any projectors $P$ and $Q$, we have
  \[
    \|P+Q\|_\infty \le 1+\|PQ\|_\infty.
  \]
\end{lemma}

\begin{proof}
  Since $P$ and $Q$ are projectors, $P+Q$ is positive
  semi-definite. Let $v$ be a unit eigenvector of $P+Q$ with
  eigenvalue $\lambda=\|P+Q\|_\infty$. Then
  \[
    \lambda^2 = \langle v, \left( P+Q \right)^2 v \rangle.
  \]
  Write $a=\|Pv\|$ and $b=\|Qv\|$. Expanding $(P+Q)^2$, we obtain
  \[
    \lambda^2 = \langle v,Pv\rangle+\langle v,Qv\rangle+\langle
    v,PQv\rangle+\langle v,QPv\rangle =
    a^2+b^2+2\operatorname{Re}\langle v,PQv\rangle.
  \]
  By Cauchy-Schwarz,
  \[
    |\langle v,PQv\rangle|\le \|Pv\|\,\|Qv\|=ab.
  \]
  Also,
  \[
    ab=|\langle Pv,Qv\rangle|=|\langle v,PQv\rangle|\le \|PQ\|_\infty.
  \]
  Therefore
  \[
    \lambda^2 \le a^2+b^2+2\|PQ\|_\infty \le
    (1+\|PQ\|_\infty)(a^2+b^2),
  \]
  where the last step uses $a^2+b^2\le 2$ and $\|PQ\|_\infty\le 1$.
  Finally,
  \[
    a^2+b^2=\langle v,(P+Q)v\rangle=\lambda,
  \]
  so
  \[
    \lambda^2\le (1+\|PQ\|_\infty)\lambda.
  \]
  Since $\lambda\ge 0$, this implies
  \[
    \|P+Q\|_\infty=\lambda\le 1+\|PQ\|_\infty.
  \]
\end{proof}

\subsection{Classically Accessible Random Oracle Model}

\begin{definition}[$q$-query CAROM strategy]
  A $q$-query CAROM strategy is an arbitrary quantum strategy that may
  retain ancillary quantum systems and apply arbitrary joint quantum
  operations between oracle calls, but makes at most $q$ adaptive calls
  whose inputs and outputs are classical.  No time, depth, locality, or
  internal-memory bound is imposed by this definition.
\end{definition}

\subsection{Quantum Interactive Proof}

\begin{definition}[Quantum interactive proof system]
  A \emph{quantum interactive proof system} $V$ consists of a
  polynomial-time quantum verifier interacting with an unbounded
  quantum prover for a fixed number of rounds.  The verifier then
  outputs either accept or reject.
\end{definition}

\begin{definition}[Value of a quantum interactive proof system]
  Let $V$ be a quantum interactive proof system. Its \emph{value},
  denoted $\omega(V)$, is the maximum acceptance probability
  achievable by any quantum prover interacting with the verifier of
  $V$.
\end{definition}

\begin{definition}[Parallel repetition]
  Let $V^1,\ldots,V^k$ be quantum interactive proof systems. Their
  \emph{parallel repetition}
  \[
    V^1\otimes\cdots\otimes V^k
  \]
  is the quantum interactive proof system obtained by running all $k$
  verifier procedures in parallel, allowing the prover to use an
  arbitrary joint strategy across the repeated executions, and
  accepting if and only if every verifier instance accepts.
\end{definition}

\begin{lemma}[Perfect parallel repetition, equation 4.53
  of~\cite{vidick2016quantum}]\label{lem:parallel-repetition}
  For any quantum interactive proof systems $V^1,\ldots,V^k$, we have
  \[
    \omega(V^1\otimes V^2\otimes\cdots\otimes V^k) =
    \omega(V^1)\cdots\omega(V^k).
  \]
  Here $\omega(V)$ denotes the value of the proof system $V$, and
  $V^1\otimes\cdots\otimes V^k$ denotes the proof system that runs the
  corresponding verifier procedures in parallel. This multiplicativity
  follows from the perfect parallel repetition theorem for quantum
  interactive proofs.
\end{lemma}

\subsection{Cryptography and Random Oracles}


\begin{definition}[Memory Token in the CAROM]
  Let $\lambda \in \mathbb{N}$ be the security parameter, and let
  \(\cO\) be a classically accessible random oracle with output domain
  \(\bit^{\lambda}\).  A \emph{memory token scheme} in the CAROM is a
  pair of quantum algorithms
  $(\mathsf{Token.Gen}^{\cO},\mathsf{Token.Eval}^{\cO})$ satisfying
  the following.

  \begin{enumerate}
  \item \textbf{Syntax.}
    \begin{enumerate}
    \item
      $\ket{\mathsf{tk}} \leftarrow
      \mathsf{Token.Gen}^{\cO}(1^\lambda,m_0,m_1)$ is a quantum
      algorithm that takes as input a security parameter and two
      messages $m_0,m_1 \in \bit^{\lambda}$, and outputs a quantum
      state $\ket{\mathsf{tk}}$, called a \emph{memory token}.
    \item
      $m \leftarrow \mathsf{Token.Eval}^{\cO}(\ket{\mathsf{tk}},b)$ is
      a quantum algorithm that takes as input a quantum state
      $\ket{\mathsf{tk}}$ and a bit $b \in \bit$, and outputs a
      message $m$.
    \end{enumerate}

  \item \textbf{Completeness.}  For every security parameter
    $\lambda$, every pair of messages $m_0,m_1 \in \bit^{\lambda}$,
    and every bit $b \in \bit$,
    \[
      \Pr\left[ \mathsf{Token.Eval}^{\cO}(\ket{\mathsf{tk}},b)=m_b :
        \ket{\mathsf{tk}} \leftarrow
        \mathsf{Token.Gen}^{\cO}(1^\lambda,m_0,m_1) \right] = 1 -
      \negl(\lambda),
    \]
    where the probability is over the choice of \(\cO\), the internal
    randomness of both algorithms and all measurement outcomes.

  \item \textbf{Efficiency.}  If $n$ and $\ell$ are polynomial in
    $\lambda$, both algorithms run in polynomial time.
  \end{enumerate}
\end{definition}

We now define simulation security for OTM in the CAROM.

\begin{definition}[One-Time Memory with Query-Bounded Simulation Security]
  A memory token scheme
  \[
    (\mathsf{Token.Gen}^{\cO},\mathsf{Token.Eval}^{\cO})
  \]
  is called a \emph{one-time memory in the CAROM with simulation
    security} if, in addition to the properties above, it satisfies
  the following security condition for the stated parameter-dependent
  error bound $\varepsilon(\lambda,n,\ell,q)$.

  For every $q$-query CAROM strategy $\cA$, there exists a simulator
  $\mathsf{Sim}$ such that, for every security parameter $\lambda$,
  every pair of messages $m_0,m_1 \in \bit^{\lambda}$,
  \[
    \left| \Pr\left[
        \cA^{\cO}(\mathsf{Token.Gen}^{\cO}(1^\lambda,m_0,m_1)) =1
      \right] - \Pr\left[ \cA^{\cO_{\Sim}}\left(\Sim^{g^{m_0, m_1}}
          \left( 1^{\lambda} \right)\right)=1 \right] \right| \le
    \varepsilon(\lambda,n,\ell,q),
  \]
  where $\Sim$ makes at most one classical query to
  \(g^{m_0, m_1} : b \mapsto m_b\), \(\cO_{\Sim}\) is the simulated
  oracle controlled by \(\Sim\), and the probability is over the choice
  of the uniformly random oracle $\cO$, the internal randomness of all
  algorithms, and all measurement outcomes.
  The quantitative error need not be negligible for arbitrary parameter
  choices; the construction below gives its explicit value.  When
  $\cA$, $n$, $\ell$, and $q$ are polynomial in $\lambda$, the simulator
  is efficient and this definition implies the usual computational CAROM
  statement.
\end{definition}

\section{One-time Memory in the CAROM}
\label{sec:otm}

In this section we construct and prove the OTM described in the
introduction, working throughout with $n$ BB84 words of length $\ell$,
where $n=n(\lambda)\geq 1$ and $\ell=\ell(\lambda)\geq 1$ are
functions of $\lambda$, and the two messages remain in
$\bit^\lambda$.  The token contains $n\ell$ qubits and
$n\lambda+2\lambda$ classical bits.  The correctness error is at most
$n2^{-\lambda}$.  The security analysis gives the explicit
query-dependent bound in~\Cref{thm:master-thm-for-sec}; it is
information-theoretic in the internal quantum strategy.  Generation
and evaluation are efficient when $n$ and $\ell$ are polynomial in
$\lambda$, and the simulator has polynomial overhead for
polynomial-time adversaries with polynomial $q$.

\subsection{Construction}

We use two domain-separated restrictions of one uniformly random oracle:
\[
  \cO_{\mathrm{tag}}:[n]\times\bit^\ell\times\bit\to\bit^\lambda,
  \qquad
  \cO_{\mathrm{mask}}:[n]\times\bit^\ell\times\bit^n\to\bit^\lambda.
\]
Equivalently, the two restrictions may be sampled as independent random
functions.  The first restriction produces public tags, while the second
produces the message masks.  Write
$\theta=(\theta_1,\ldots,\theta_n)$ and
$S_b=\{i\in[n]:\theta_i=b\}$.  All XORs below are bit-wise, and the XOR
of an empty set is $0$.

For a basis bit $\beta$, write
\[
  \ket{x}_{\beta}:=(H^{\beta})^{\otimes\ell}\ket{x},
\]
so basis $0$ is the computational (Z) basis and basis $1$ is the
Hadamard (X) basis.  The protocol is given in~\Cref{alg:otm-protocol}.
The sender encodes independent words in random bases, authenticates each
word with a tag, and masks each message with the XOR of the mask-oracle
values indexed by the corresponding words and the full basis pattern.
The receiver measures all registers in the basis determined by its
selection bit and tests each outcome against the public tags; these
checks recover the full basis pattern, enabling the receiver to
reconstruct the selected mask and recover the chosen message, except
with the error probability bounded in~\Cref{thm:otm-correctness}.

\begin{algorithm}[t]
  \caption{One-time memory protocol\label{alg:otm-protocol}}
  \begin{minipage}[t]{0.49\textwidth}
    \textbf{Sender: $\mathsf{Token.Gen}^{\cO}(1^\lambda,m_0,m_1)$}
    \begin{algorithmic}[1]
      \State Sample $x_i\getsr\bit^{\ell}$ and
      $\theta_i\getsr\bit$ for all $i\in[n]$.
      \State Prepare $\ket{x_i}_{\theta_i}
        =(H^{\theta_i})^{\otimes\ell}\ket{x_i}$ for every $i\in[n]$.
      \State Set
      \[
        y_b\gets\bigoplus_{i\in S_b}
          \cO_{\mathrm{mask}}(i,x_i,\theta_i),
        \qquad b\in\bit.
      \]
      \State Set $c_b\gets m_b\oplus y_b$ for $b\in\bit$.
      \State Set
      $\tau_i\gets\cO_{\mathrm{tag}}(i,x_i,\theta_i)$ for all $i\in[n]$.
      \State Let
      $\ket{\psi}\gets\bigotimes_{i=1}^{n}\ket{x_i}_{\theta_i}$
      and $c\gets((\tau_i)_{i=1}^{n},c_0,c_1)$.
      \State Output $\ket{\tk}=(\ket{\psi},c)$.
    \end{algorithmic}
  \end{minipage}
  \hfill
  \begin{minipage}[t]{0.49\textwidth}
    \textbf{Receiver: $\mathsf{Token.Eval}^{\cO}(\ket{\tk},b)$}
    \begin{algorithmic}[1]
      \State Parse $\ket{\tk}$ as $\ket{\psi}$ and
      $c=((\tau_i)_{i=1}^{n},c_0,c_1)$.
      \State Measure every register in basis $b$: the computational
      (Z) basis for $b=0$ and the Hadamard (X) basis for $b=1$.
      \State Let $x'_1,\ldots,x'_n\in\bit^{\ell}$ be the outcomes.
      \State Set
      \[
        S\gets\{i\in[n]:\cO_{\mathrm{tag}}(i,x'_i,b)=\tau_i\}
      \]
      and set
      \[
        \theta'_i=\begin{cases}
          b,&i\in S,\\
          1-b,&\text{otherwise}.
        \end{cases}
      \]
      \State Output
      \[
        c_b\oplus\bigoplus_{i\in S}
        \cO_{\mathrm{mask}}(i,x'_i,\theta').
      \]
    \end{algorithmic}
  \end{minipage}
\end{algorithm}

\begin{theorem}[Correctness]
  \label{thm:otm-correctness}
  For every $m_0,m_1\in\bit^\lambda$ and every $b\in\bit$,
  \[
    \Pr\left[\mathsf{Token.Eval}^{\cO}\left(
      \mathsf{Token.Gen}^{\cO}(1^\lambda,m_0,m_1),b\right)=m_b\right]
    \geq 1-n2^{-\lambda}.
  \]
\end{theorem}

\begin{proof}
  Fix $b\in\bit$.  If $\theta_i=b$, measurement in basis $b$ recovers
  $x_i$ exactly, so the tag test accepts.  If $\theta_i=1-b$, the tag
  query $(i,x'_i,b)$ differs from the sender's tag point
  $(i,x_i,\theta_i)$ in the basis input, and hence matches $\tau_i$
  with probability $2^{-\lambda}$. A union bound
  gives the stated error.  When no false match occurs, $S=S_b$,
  $\theta'=\theta$, and $x'_i=x_i$ for every $i\in S$.  The output is
  therefore

  \[
    c_b\oplus\bigoplus_{i\in S_b}\cO_{\mathrm{mask}}(i,x_i,\theta)
    =(m_b\oplus y_b)\oplus y_b=m_b.
  \]
\end{proof}

\subsection{Security events and the full-pattern condition}

Intuitively, the message $m_{1-b}$ remains secret as long as some word
$x_i$ with $\theta_i=1-b$ has not been revealed to the adversary: the
corresponding mask contribution is then a fresh uniform value that
perfectly hides the message.  The events below make this precise.

For each $i\in[n]$, let
\[
  M_i=(i,x_i,\theta)
\]
be the correct mask point.  Let $C_b$ be the event that the adversary's
classical transcript contains every $M_i$ with $i\in S_b$.  To express the
local information needed by this event, let $F_i$ be the event that the
transcript contains either a tag query $(i,x_i,\theta_i)$ or a mask
query $(i,x_i,\Theta)$ with $\Theta_i=\theta_i$.  Thus every exact
full-pattern mask query for word $i$ implies $F_i$.

The key implication is
\begin{equation}
  C_0\wedge C_1\ \Longrightarrow\ \bigcap_{i=1}^{n}F_i
  \qquad\text{whenever }S_0,S_1\neq\varnothing.
  \label{eq:full-pattern-implies-local}
\end{equation}
If one local event $F_i$ fails, at least one mask contribution remains
unqueried, so the corresponding XOR is a uniform one-time pad even if all
other words have been recovered.  Partial local recovery therefore does
not need a separate leakage analysis.  The exceptional event
\[
  E_{\mathrm{empty}}=\{S_0=\varnothing\}\cup\{S_1=\varnothing\}
\]
has probability $2^{1-n}$.

\subsection{Query-bounded recovery}
\label{subsec:proof-of-hardness}

The core difficulty in bounding~\eqref{eq:full-pattern-implies-local}
is that the oracle queries and the published ciphertexts $c_0,c_1$
correlate all words.  We therefore proceed in two steps.  First, we
define an idealized single-word game in which the oracle is replaced by
an equality check and all side information is removed, and bound its
success probability via a list-recovery argument.  Second, we reduce
completion of both branches in the real experiment to the $n$-fold
parallel repetition of this game, losing a factor of $n+1$ for the
unknown size of the second basis class.

The single-word game records the local information obtainable through
equality tests.
The adversary may otherwise apply an arbitrary quantum strategy.

\begin{algorithm}[t]
  \caption{Single-word recovery game with equality-test oracle
    $\cG_{\mathrm{single}}(q,\lambda,\ell)$\label{alg:sec-game-single}}
  \begin{algorithmic}[1]
    \State Sample $x\getsr\bit^\ell$ and $\theta\getsr\bit$ and send
    $\ket{x}_{\theta}=(H^\theta)^{\otimes\ell}\ket{x}$ to the adversary.
    \State For at most $q$ classical queries $(y,\xi)$, return $1$ if
    $y=x$ and $\xi=\theta$ and $0$ otherwise.
    \State The adversary outputs $z\in\bit^\ell$ and wins iff it either
    queried $(x,\theta)$ or outputs $z=x$.
  \end{algorithmic}
\end{algorithm}

\begin{lemma}[List-recovery reduction]
  \label{lem:list-recovery-reduction}
  Let $\cA$ win $\cG_{\mathrm{single}}(q,\lambda,\ell)$ with probability
  at least $\varepsilon$.  Then there exists an oracle-free quantum
  strategy that, receiving
  only $\ket{x}_{\theta}$, outputs a list $L\subseteq\bit^\ell$ of size at
  most $q+1$ such that $\Pr[x\in L]\geq\varepsilon$.
\end{lemma}

\begin{proof}
  Run $\cA$, answer $0$ for every equality-test oracle query, and
  output the first component of each query together with $\cA$'s final
  output.  Until the first correct query, the simulated answers are
  distributed exactly as in the game, so every successful query or
  final output appears in the list.
\end{proof}

\begin{lemma}[BB84 list-recovery bound]
  \label{lem:bb84-list-recovery-bound}
  Let a quantum strategy receive only $\ket{x}_{\theta}$ and output a list
  of at most $L\leq 2^\ell$ candidates.  Then
  \[
    \Pr[x\in L]\leq \frac12+\frac{L}{2^{\ell/2+1}}.
  \]
\end{lemma}

\begin{proof}
  Average over the hidden basis and write
  \[
    \rho_x=\frac12\ket{x}\bra{x}+
      \frac12H^{\otimes\ell}\ket{x}\bra{x}H^{\otimes\ell}.
  \]
  For a fixed list $S$ of size $L$, let
  $P_S=\sum_{x\in S}\ket{x}\bra{x}$.  The corresponding contribution is
  \[
    K_S=\frac12\left(P_S+H^{\otimes\ell}P_SH^{\otimes\ell}\right).
  \]
  The projector-sum bound (\Cref{lem:projector-sum}) and Schatten
  norms give
  \[
    \|K_S\|_\infty
      \leq\frac12\left(1+\|P_SH^{\otimes\ell}P_S\|_\infty\right)
      \leq\frac12\left(1+L2^{-\ell/2}\right).
  \]
  Summing the POVM effects over all list outcomes proves the claim.
\end{proof}

The bound is essentially tight: for $\ell=1$ and $L=1$, the optimal
(Helstrom) measurement discriminating the two possible states succeeds
with probability $\frac12+\frac{1}{2\sqrt2}$, so the $2^{-\ell/2}$
scaling cannot be improved in general.

Define
\[
  p(q,\ell):=\frac12+\frac{q+1}{2^{\ell/2+1}}.
\]
The list-recovery reduction and the BB84 bound imply that the success
probability of one local game is at most $p(q,\ell)$.  Each local game
can be viewed as a $q$-round quantum interactive proof system: the
  verifier samples $x,\theta$, sends $\ket{x}_{\theta}$, answers
  equality queries, and accepts iff the adversary outputs $x$.  The
  $n$-fold
local game is then the parallel repetition of $n$ independent copies
of this proof system, against an adversary that may use an arbitrary
joint strategy.  Perfect parallel repetition
(\Cref{lem:parallel-repetition}), while allowing the adversary $q$
queries in every coordinate (a relaxation of its total query budget),
gives
\begin{equation}
  \Pr[\text{all $n$ local tests succeed}]
  \leq p(q,\ell)^n.
  \label{eq:local-direct-product}
\end{equation}
Here the probability is taken in this independent product game, where
the adversary may use an arbitrary joint quantum strategy across the
$n$ coordinates.  The real experiment is embedded into this game by
the reduction of~\Cref{lem:direct-product-hardness} below: each
classical oracle query of the adversary is forwarded to one
coordinate, so granting $q$ queries per coordinate is a valid
relaxation of the total budget $q$, and no product assumption on the
  adversary's behavior is needed.  The bound is information-theoretic;
  no additional restriction on the adversary is needed.

\begin{lemma}[Full-pattern completion bound]
  \label{lem:direct-product-hardness}
  Fix messages $m_0,m_1\in\bit^\lambda$ and consider the following
  \emph{random-pad experiment}.  The challenger samples
  $x_i\getsr\bit^\ell$ and $\theta_i\getsr\bit$, prepares the token
  state $\bigotimes_{i=1}^{n}\ket{x_i}_{\theta_i}$, and samples
  independent uniform $\tau_i,c_0,c_1\in\bit^\lambda$.  It sets
  $\cO_{\mathrm{tag}}(i,x_i,\theta_i)=\tau_i$ for every $i\in[n]$, and
  for each branch $b\in\bit$ with $S_b\neq\varnothing$ it samples the
  outputs of $\cO_{\mathrm{mask}}$ at the exact mask points
  $\{M_i:i\in S_b\}$ uniformly subject to
  \[
    \bigoplus_{i\in S_b}\cO_{\mathrm{mask}}(i,x_i,\theta)
      =c_b\oplus m_b,
  \]
  while all remaining outputs of both oracles are uniform and
  independent.  The adversary receives the BB84 registers together
  with the classical values $(\tau_i)_{i=1}^{n},c_0,c_1$ and arbitrary
  auxiliary input, and may make at most $q$ adaptive classical queries
  to $\cO_{\mathrm{tag}}$ and $\cO_{\mathrm{mask}}$, with arbitrary
  quantum computation between queries.  Then
  \[
    \Pr[C_0\wedge C_1]\le (n+1)\,p(q,\ell)^n,
  \]
  where the probability is over the sampling above and the adversary's
  strategy.
\end{lemma}

\begin{proof}
  We prove the bound in three steps: we first reduce the random-pad
  experiment to an \emph{equality-test experiment}, in which the
  adversary receives the BB84 registers
  $\bigotimes_{i=1}^{n}\ket{x_i}_{\theta_i}$ and, for each word $i$,
  an equality-test oracle that answers a query $(y,\xi)$ by whether
  $(y,\xi)=(x_i,\theta_i)$; the adversary wins if it tests the correct
  pair for every $i$.  We then view this
  experiment as the parallel repetition of $n$ independent single-word
  local games, each with its own independent equality-test oracle, and
  finally bound each local game separately with the list-recovery
  reduction.

  \emph{Step 1: reducing to the equality-test experiment.}  The
  reduction holds the product state
  $\bigotimes_{i=1}^{n}\ket{x_i}_{\theta_i}$, samples independent
  uniform $\tau_1,\ldots,\tau_n$ and $c_0,c_1$, guesses $k=|S_1|$
  uniformly from $0,1,\ldots,n$, and runs the adversary of the
  random-pad experiment on the state
  $\bigotimes_{i=1}^{n}\ket{x_i}_{\theta_i}$ and the values
  $\tau_1,\ldots,\tau_n,c_0,c_1$, answering its oracle queries as
  follows.

  \emph{Tag queries as equality tests.}  In the experiment a tag
  response is independently uniform at every point other than
  $(i,x_i,\theta_i)$; hence the only information a tag query can yield
  about $(x_i,\theta_i)$ is whether the queried pair is correct.  The
  reduction exploits exactly this structure: for a tag query
  $(i,y,\xi)$ it forwards $(y,\xi)$ to the $i$-th local game, whose
  answer is an equality test, and returns $\tau_i$ on a positive test
  and a fresh uniform string otherwise.  This reproduces the
  distribution of $\cO_{\mathrm{tag}}$ exactly.

  \emph{Mask queries.}  For a mask query $(i,y,\Theta)$ the reduction
  forwards the projected pair $(y,\Theta_i)$ to the $i$-th local game.
  If the test fails, the queried point is not an exact mask point, and
  the reduction answers with a fresh uniform value, recorded for
  consistency.  If the test passes, the reduction records the query
  with its full pattern $\Theta$ and answers as follows.

  Only patterns with exactly $k$ ones can carry the XOR constraints,
  since under the correct guess $\theta$ itself has $k$ ones; a query
  whose pattern has any other Hamming weight receives a fresh uniform
  value.  For a pattern $\Theta$ with $k$ ones, the reduction counts
  the distinct passing points $(i,x_i,\Theta)$, separately for the two
  branches $\Theta_i=1$ and $\Theta_i=0$.  On the $k$-th passing point
  with $\Theta_i=1$ it returns the value that closes
  \[
    \bigoplus_{j:\Theta_j=1}\cO_{\mathrm{mask}}(j,x_j,\Theta)
      =c_1\oplus m_1,
  \]
  and on the $(n-k)$-th passing point with $\Theta_i=0$ it closes the
  branch-$0$ constraint symmetrically; every other passing point
  receives a fresh uniform value, recorded for consistency.  For any
  pattern other than the true $\theta$, the tests can pass on at most
  $k-1$ points of the $1$-branch and at most $n-k-1$ points of the
  $0$-branch, so the constraints are only ever closed for
  $\Theta=\theta$.

  The local games certify only the projected pair; the full pattern of
  each query is kept in the reduction's record.  The one quantity the
  reduction cannot learn from the games is the number of ones in
  $\theta$, which is why it guesses $k$.  Conditioned on the correct
  guess, this lazy sampling is distributed exactly as the random-pad
  experiment.  On the event $C_0\wedge C_1$ the adversary has queried
  the exact mask point $M_i$ for every $i\in S_0\cup S_1=[n]$, so the
  forwarded test $(x_i,\theta_i)$ passes in every coordinate and the
  reduction wins the equality-test experiment.  As the correct guess
  is made with probability $1/(n+1)$, we obtain
  \[
    \Pr[\text{the reduction wins the equality-test experiment}]\ge
      \frac{\Pr[C_0\wedge C_1]}{n+1}.
  \]

  \emph{Step 2: to the product of local recoveries.}  The
  equality-test experiment is exactly the parallel repetition of the
  $n$ independent single-word games of~\Cref{alg:sec-game-single}:
  each word $i$ contributes a local recovery game with its own
  independent $x_i$ and $\theta_i$, where each equality-test call is
  one round of interaction between the prover (the adversary) and the
  verifier (the challenger), and the joint
  adversary may use an arbitrary entangled strategy across the
  coordinates.  By perfect parallel repetition
  (\Cref{lem:parallel-repetition}), the probability of winning all $n$
  local games is the product of the single-word values $\omega(V_i)$.

  \emph{Step 3: bounding each local game separately.}  For each local
  game $i$ separately, the list-recovery reduction
  of~\Cref{lem:list-recovery-reduction} removes the equality-test
  oracle: any strategy that wins the single-word game with probability
  $\omega(V_i)$ yields an oracle-free strategy, receiving only
  $\ket{x_i}_{\theta_i}$, that outputs a list of at most $q+1$
  candidates containing $x_i$ with probability at least
  $\omega(V_i)$.  By the BB84
  list-recovery bound of~\Cref{lem:bb84-list-recovery-bound},
  $\omega(V_i)\le p(q,\ell)$ for every $i$.  Combining the steps gives
  \[
    \Pr[C_0\wedge C_1]
      \le (n+1)\prod_{i=1}^{n}\omega(V_i)
      \le (n+1)\,p(q,\ell)^n,
  \]
  as recorded in~\eqref{eq:local-direct-product}.
\end{proof}

\subsection{Simulation security}

\begin{theorem}[Query-bounded simulation security]
  \label{thm:master-thm-for-sec}
  For every adversary making at most $q$ adaptive classical queries to the
  classically accessible random oracle, with arbitrary quantum computation
  between queries, the protocol in~\Cref{alg:otm-protocol} admits a
  simulator making at most one query to $g^{m_0,m_1}$.  Its distinguishing
  advantage is at most
  \[
    \delta(n,\ell,q)=2^{1-n}+(n+1)
      \left(\frac12+\frac{q+1}{2^{\ell/2+1}}\right)^n.
  \]
  When the adversary, $n$, $\ell$, and $q$ are polynomial in $\lambda$,
  the simulator is efficient and this implies the computational CAROM
  statement.
\end{theorem}

\begin{proof}
  The hybrids below show that the real experiment and the ideal one
  differ only if the adversary completes both branches; the probability
  of this event is then bounded
  by~\Cref{lem:direct-product-hardness}.

  Let $E_{\mathrm{empty}}$ be as above.  In Hybrid~0, generate the token
  honestly.  In Hybrid~1, sample $c_0,c_1$ uniformly and, conditioned on
  $\neg E_{\mathrm{empty}}$, program one mask-oracle value per branch as
  the dependent value making
  \[
    \bigoplus_{i\in S_b}\cO_{\mathrm{mask}}(i,x_i,\theta)
      =c_b\oplus m_b.
  \]
  All other values are sampled lazily and consistently.  On
  $\neg E_{\mathrm{empty}}$, Hybrid~0 and Hybrid~1 have the same joint
  distribution; hence their statistical distance is at most
  $\Pr[E_{\mathrm{empty}}]=2^{1-n}$.

  The simulator samples $x_i,\theta_i$, prepares the BB84 registers, and
  samples all $\tau_i,c_0,c_1$ uniformly.  Because every query is
  classical, the simulator can inspect the transcript and detect the
  first query that completes a branch.  It maintains a table of every
  query--answer pair, initially containing the tag points
  $(i,x_i,\theta_i)\mapsto\tau_i$.  On a repeated query it returns the
  stored answer.  On a fresh exact mask query
  $(i,x_i,\theta)$, it checks whether every other mask point in the same
  branch has already been queried.  If so and the ideal functionality has
  not yet been queried, it obtains $m_{\theta_i}$ and returns the unique
  value satisfying the branch XOR constraint.  Otherwise it returns a
  fresh uniform value.  Every returned value, including a programmed one,
  is stored in the table.

  As long as at least one branch mask point remains unqueried, this table
  has exactly the Hybrid~1 distribution.  A discrepancy can occur only if
  both $C_0$ and $C_1$ occur: the first completion consumes the single
  ideal query, while the second completion would require the other message.
  By~\Cref{lem:direct-product-hardness}, this event has probability at
  most $(n+1)p(q,\ell)^n$.  Together with the empty-basis distance
  $2^{1-n}$ between Hybrids~0 and 1, this proves the stated bound.
\end{proof}

\section*{AI Disclaimer}
The main ideas and proofs are comed up by the authors without assistance of AI. We used OpenAI Codex to assist with drafting and editing portions of the
introduction and technical overview, as well as with formatting checks.
The authors verified the correctness and originality of all content,
including the references. The authors take full responsibility of the paper.

\bibliographystyle{alpha} \bibliography{abbrev3,crypto,bibliography}

@string{ieee =                  {IEEE}}

@string{acm =                   "Association for Computing Machinery"}

@string{acm =                   "{ACM}"}

@article{vidick2016quantum,
  title={Quantum proofs},
  author={Vidick, Thomas and Watrous, John},
  journal={Foundations and Trends in Theoretical Computer Science},
  volume={11},
  number={1-2},
  pages={1--215},
  year={2016},
  publisher={Emerald Publishing Limited}
}

@ARTICLE{Win99gentle,
  author={Winter, A.},
  journal={IEEE Transactions on Information Theory}, 
  title={Coding theorem and strong converse for quantum channels}, 
  year={1999},
  volume={45},
  number={7},
  pages={2481-2485},
  doi={10.1109/18.796385}}

@misc{stambler2025geometric,
  author = {Lev Stambler},
  title = {Information Theoretic One-Time Programs from Geometrically
    Local {QNC0} Adversaries},
  year = {2025},
  eprint = {2503.22016},
  archivePrefix = {arXiv},
  primaryClass = {quant-ph},
  url = {https://arxiv.org/abs/2503.22016}
}

@misc{stambler2025qotm,
  author = {Lev Stambler},
  title = {Quantum One-Time Memories from Stateless Hardware, Random
    Access Codes, and Simple Nonconvex Optimization},
  year = {2025},
  eprint = {2501.04168},
  archivePrefix = {arXiv},
  primaryClass = {quant-ph},
  url = {https://arxiv.org/abs/2501.04168}
}

@misc{stambler2025semquantum,
  author = {Lev Stambler},
  title = {Cryptography without Long-Term Quantum Memory and Global
    Entanglement: Classical Setups for One-Time Programs, Copy
    Protection, and Stateful Obfuscation},
  year = {2025},
  eprint = {2504.21842},
  archivePrefix = {arXiv},
  primaryClass = {quant-ph},
  url = {https://arxiv.org/abs/2504.21842}
}

@misc{stambler2026simple,
  author = {Lev Stambler},
  title = {Towards Simple and Useful One-Time Programs in the Quantum
    Random Oracle Model},
  year = {2026},
  eprint = {2601.13258},
  archivePrefix = {arXiv},
  primaryClass = {quant-ph},
  url = {https://arxiv.org/abs/2601.13258}
}

@article{wiesner1983conjugate,
  title={Conjugate coding},
  author={Wiesner, Stephen},
  journal={ACM Sigact News},
  volume={15},
  number={1},
  pages={78--88},
  year={1983},
  publisher={ACM New York, NY, USA}
}

@article{CGLZ19,
  title={Cryptography with disposable backdoors},
  author={Chung, Kai-Min and Georgiou, Marios and Lai, Ching-Yi and Zikas, Vassilis},
  journal={Cryptography},
  volume={3},
  number={3},
  pages={22},
  year={2019},
  publisher={MDPI}
}

\end{document}